\documentclass[11pt]{article}
\usepackage{amsmath,amssymb}
\usepackage{url}

\newtheorem{theorem}{\bf Theorem}[section]

\newtheorem{definition}[theorem]{Definition}
\newtheorem{definitions}[theorem]{Definitions}
\newtheorem{corollary}[theorem]{Corollary}
\newtheorem{lemma}[theorem]{Lemma}

\newtheorem{proposition}[theorem]{Proposition}

\newtheorem{varnote}[theorem]{Note}

\newenvironment{proof}{
  \noindent\textbf{Proof}\ }{\hspace*{\fill}
  \begin{math}\Box\end{math}\medskip}
\newenvironment{proof*}[1]{
  \noindent\textbf{#1\ }}{\hspace*{\fill}
  \begin{math}\Box\end{math}\medskip}

\newtheorem{varremark}[theorem]{Remark}
\newenvironment{remark}{\begin{varremark}\em}{\em\end{varremark}}
\newtheorem{problem}{Problem}

\newcommand{\tL}{L}

\newcommand{\N}{{\mathbb N}}

\newcommand{\C}{{\mathbb C}}

\newcommand{\diag}{{\rm diag}}

\newcommand{\rank}{{\rm rank }}

\newcommand{\QED}{\rule{2.3mm}{2.3mm}}

\newcommand{\bz}{{\bf z}}
\newcommand{\bx}{{\bf x}}
\newcommand{\vz}{{\bf z}}

\newcommand{\para}[1]{\left( #1 \right)}
\newcommand{\set}[1]{\left\{ #1 \right\}}
\newcommand{\z}{\mathbf{z}}
\newcommand{\x}{\mathbf{x}}
\newcommand{\e}{\mathbf{e}}
\newcommand{\CoeffVec}{W}
\newcommand{\vecz}{{\bf z}}
\newcommand{\veczconj}{\bar{\bf z}}
\newcommand{\newvecx}{\vecz'}
\newcommand{\oldvecx}{\vecz}
\newcommand{\FullVandermonde}{V}

\begin{document}

\title{Over-constrained Weierstrass iteration and the nearest consistent
system }
\author{Olivier Ruatta\thanks{Universit\'e de Limoges},\;\;
{Mark Sciabica,\;\;
Agnes Szanto}\thanks{North Carolina State University, Raleigh, NC. This research was partly supported by NSF grants
CCR-0306406, CCR-0347506, DMS-0532140 and CCF-1217557.}}

\date{\today}

 \pagestyle{myheadings}
\markright{O. Ruatta, M. Sciabica, A. Szanto }
\maketitle

\begin{abstract}
We propose a generalization of the Weierstrass iteration for
over-constrained systems of equations and we prove that the proposed
method is the Gauss-Newton iteration to find the nearest system which
has at least $k$ common roots and which is obtained via a perturbation
of prescribed structure. In the univariate case we show the connection
of our method to the optimization problem formulated by Karmarkar and
Lakshman for the nearest GCD. In the multivariate case we generalize
the expressions of Karmarkar and Lakshman, and give explicitly  
several iteration functions to compute the optimum.
 The arithmetic complexity of the
 iterations is detailed.\\
 
 {\bf Keywords:} Overdetermined systems, nearest consistent system, Weierstrass Durand Kerner method

\end{abstract}

\section{Introduction}

In many physical and engineering applications one needs to solve
over-constrained   systems of equations, i.e.\! systems  with more
equations than unknowns,  such that the existence of the solutions is
guaranteed by some underlying physical property. However,  the input
system may be given only with limited accuracy due to measurement or
rounding error, and thus the actual input
 may be inconsistent.\\

The work presented in this paper is concerned with the question of
finding the ``nearest'' system with at least $k$ distinct common roots
over $\C$.  We introduce a generalization of the Gauss-Weierstrass
method \cite{Weier03,ORthese}. In the univariate case, the proposed iterative
method allows computation of  the nearest GCD of given degree, and is
closely related to the formula of Karmarkar-Lakshman  for the distance
to the set of systems with at least $k$ common roots \cite{KaLa96,KarLak1998}.
We show how to extend the iterative method to over-constrained systems
of analytic functions. Using this extended construction we generalize
the Karmarkar-Lakshman formula to the  multivariate case. \\

More precisely, in the univariate case the problem we address in the
paper is the following:

\begin{problem} \label{prob1}
Given $f,g \in \C[x]$ and $k \in \N$, find a polynomial $h$ of degree
$k$ such that there exist  polynomials $\tilde{f},\tilde{g} \in \C[x]$
such that  $h$ divides both $\tilde{f}$ and $\tilde{g}$, and
$f\!-\!\tilde{f}$ and $g\!-\!\tilde{g}$ have prescribed supports and
minimal 2-norms.
\end{problem}

 The method proposed here is based on a generalization of the
so-called Weierstrass method (also called Durand-Kerner method
\cite{Durand60,Kerner,Durand} or Dochev method \cite{Frommer1,SAK00})
introduced in \cite{Weier03}  and successively generalized in 
\cite{Bellido3,OR01,ORthese,MR02} (for a survey on the history see \cite{Pan97}). Our  first contribution in the univariate case is to  show a link  between the Weierstrass  method and  the formulation
of Karmarkar and Lakshman in \cite{KarLak1998} using Lagrange interpolation (see Theorem \ref{WmapProp}). The second contribution is an explicit formula 
for the Gauss-Newton iteration to find 
the optimum, which is derived from our expressions for the gradient of the norm 
square function (see Theorem \ref{DefWeieriteration}).\\

Next we present the extension of our results to the multivariate case. The problem we address is as follows:

 \begin{problem} \label{prob2} Given an  analytic function $\vec{f}=(f_1, \ldots, f_N)
 :\C^n\rightarrow \C^N$, $N>n$, and $k>0$, find perturbations $p_1, \ldots, p_N$ from a given finite dimensional vector space ${\mathcal P}$ of analytic functions together with $k$ distinct points ${\bf z}_1, \ldots, {\bf z}_k\in \C^n$,
 such that $f_1-p_1, \ldots, f_N-p_N$ vanishes on ${\bf z}_1, \ldots, {\bf z}_k$  
  and $\|p_1\|_2^2+\cdots +\|p_N\|_2^2$ is minimal.
\end{problem}

One of the main results of the paper is a generalization of the  formula of Karmarkar
and Lakshman in \cite{KarLak1998} for the univariate nearest GCD to the multivariate case. Using a generalization of the Lagrange interpolation we were able to express the distance of our input system to the set of systems which have at least $k$ complex roots as an optimization problem on the $k$-tuples of points in $\C^n$ (see Theorem \ref{MLagTheorem}). 
The other  main result of the paper is an explicit formulation of the Gauss-Newton iteration applied to our optimization formulation to  solve Problem \ref{prob2}.\\

Finally, we give a simplified version of the iteration, which might be of independent interest. 
Analogously to the classical Weierstrass map, we use the multivariate Lagrange interpolation polynomials in each Gauss-Newton iteration step to transform the Jacobian matrix to a block diagonal 
matrix. As a consequence, we get a simple component-wise formula for the iteration function. We show that using the simplified  method, the complexity of computing each iterate is improved compared to the non-simplified versions:   the standard Gauss-Newton iteration, the quadratic iteration, or the conjugate gradient method. However, the simplified iteration will not converge to the least squares minimum, but we do give a description 
of its fixed points. As our numerical experiments indicate, the simplified method computes roots 
with the smallest residual value $\sum_{i=1}^N \sum_{j=1}^k |f_i({\bf z}_j)|^2$, compared to the non-simplified 
versions. \\

At the end of the paper we present numerical experimentations where we compare the performances of the  simplified Gauss-Newton, the standard Gauss-Newton, the quadratic iteration, and the conjugate gradient methods to compute the optimum.
 
\subsection{Related work}

The computation of the GCD is a classical problem of symbolic
computation and efficient algorithms are known to solve it (\cite{Coll}
and \cite{BrTr} for instance). The first approach to a problem similar
to Problem \ref{prob1} was proposed by Sch\"onhage in \cite{Scho85a}
where the input polynomials are known with infinite precision. Several
later approaches were proposed where the polynomials are known with a
bounded error. In \cite{EmGaLo96,CoGiTrWa95} the authors compute upper
bounds on the degree of an $\epsilon$-GCD of two numerical polynomials
using the singular values of a Sylvester resultant matrix. In
\cite{EmGaLo97}, the authors give the exact degree of the
$\epsilon$-GCD together with a certificate using a singular value
decomposition of a subresultant matrix. In
\cite{KaLa96,KarLak1998,ChCoCo98}, the authors present the problem as a
real optimization problem and propose numerical techniques in order to
solve the optimization problem.  Hitz et al. consider the nearest
polynomial with constrained or real roots in the $l^2$ and $l^\infty$
norms  in \cite{HitzKal98,HiKaLa99}. Related approaches on approximate GCD computation 
include \cite{ RUP99,Zhi2003,Zeng2003,Stetter2004,ZengDay2004,Zeng05,Cor05, KalYanZhi2006,KalYanZhi2007,Sek08,WA08a,WA08b,LiNieZhi2008,WH10,WL11,WHL12,ElGaBa12,WH13}.\\

There are two main families of approaches  in the literature to compute
the solution of multivariate near-consistent over-constrained systems.
One type of algorithm  handles over-constrained polynomial systems with
approximate coefficients by using a symbolic-numeric approach to reduce
the problem to eigenvalue computation via {\em multiplication tables}.
The first methods in  the literature using reduction to eigenvalue
problem include \cite{AuSt,YoNoTa92,MoSt95}. The existing methods to
compute the multiplication tables use resultant matrices or Gr\"obner
basis techniques, with complexity bound
 exponential in the number of variables.\\

The other type of approaches  formulate over-constrained systems as
real optimization  problems.  Here we can only list a selected subset
of the related literature.  Giusti and  Schost  in \cite{GiuSch} reduce
the problem to the solution of a univariate polynomial. Dedieu and Shub
give a heuristic predictor corrector method in \cite{DeSh}. They also
prove alpha-theory for the Gauss-Newton method in \cite{DeSh2000}.
Stetter in \cite{Stet2000} studies the conditioning properties of
near-consistent over-constrained systems. Ruatta in \cite{ORthese}
generalizes the Weierstrass iteration for over-constrained systems and
gives a heuristic predictor corrector method based on this iteration. Recently,
Hauenstein and Sottile considered certification of approximate solutions of exact overdetermined systems in \cite{HauSot12}.

\subsection{Notations}

In all that follows, $\C$ denotes the field of complex numbers, $x$ is
an indeterminate and we denote by $\bx=(x_1, \ldots x_n)$ the vector of
$n$ indeterminates for some $n\geq 1$. $\C[x]$ and $\C[\bx]$ denote the
rings of polynomials with complex coefficients in one and $n$
indeterminates, respectively. $\C[x]_m$ is the subspace of $\C[x]$
consisting of the polynomials of degree less or equal to $m \in \N$.
For $I \subset \N$ a finite set, we denote $\C[x]_I$ the set of
polynomials with support included in $I$, i.e.
\begin{equation}
\C[x]_I=\{ p \in \C[x]\; :\; p(x)= \displaystyle \sum_{i \in I} p_i
x^i, p_i \in \C \}.
\end{equation}
For $F\subset \C[\bx]$ and ${\mathcal R}\subseteq \C^n$ we denote by
${\bf V}_{\mathcal R}(F)$ the set of common roots of $F$ in ${\mathcal
R}$. We denote indifferently $\|  \|_2$ or $\| \, \|$ the $l^2$ norm of
complex vectors which we call the 2-norm. For $f \in \C[x]$ we denote
by $\| f \|$ the 2-norm of the vector of its coefficients. For $M \in
\C^{k \times m}$, $\| M \|$ denotes the 2-norm of the vector of its
entries. The 2-norm of a vector of polynomials is the 2-norm of the
vector of all their coefficients. For a matrix $M \in \C^{k \times n}$
we denote by $M^T$ its transpose matrix and $M^*$ the transpose of the
conjugate of $M$, also  called the adjoint of $M$. For $M \in \C^{k
  \times m}$ such that $\mbox{rank}(M)=k$ (or $\mbox{rank}(M)=m$),
   we denote by
 $M^\dagger=M^*(M M^*)^{-1}$ (or $M^\dagger=(M^*M)^{-1}M^*$,  respectively) its
 Moore-Penrose pseudo-inverse.

\section{Univariate case}

In this section, we present a generalization of the Weierstrass
iteration to the approximate case. First we present a version of the classical Lagrange
interpolation method which is needed for the construction of the
iterative method. Secondly, we  define a generalization of the Weierstrass map and show the link between this map
 and the distance  to the set of systems with $k$ common roots, translating  this  distance from a minimization problem on the coefficient vector of the perturbations to a minimization problem   over
 $k$-tuples of complex number.
Next we give an explicit formula for the Gauss-Newton iteration for our optimization formulation. Finally, we give a simplified version of the iteration, which has a simple coordinate-wise iteration
function with improved  complexity.

\subsection{Generalized Lagrange interpolation}

In this subsection we introduce an optimization problem which
generalizes the classical Lagrange interpolation problem and we give a
solution to
this problem using Moore-Penrose pseudo-inverses. \\

\noindent {\bf Problem} \; [Generalized Lagrange interpolation] { \it
  Consider distinct complex numbers $z_1,\ldots,z_k \in \C$ and some
arbitrary complex numbers $f_1, \ldots,f_k \in \C$. Fix $I\subset \N$
such that $|I|\geq k$. The generalized Lagrange interpolation problem
consists of finding the minimal 2-norm  polynomial $F \in \C[x]_{I}$
with support $I$ that  satisfies:}
\begin{eqnarray}\label{conditions1}
F(z_i)=f_i \mbox{ for } i=1, \ldots, k.
\end{eqnarray}

 We will need the following definition:
\begin{definitions} Let $I=\set{i_1,\ldots,i_p} \subset \N$ such that  $p \geq
k$.
  \begin{itemize}

  \item[$\bullet$] Let $\z=(z_1,\ldots,z_k) \in \C^k$. We
      define the Vandermonde matrix associated with $\z$ and $I$ as
      following matrix of size $k \times p$:
      \begin{equation} \label{Vandermonde}
        V_{I}(\z) := \para{ \begin{array}{ccc} z_1^{i_1} & \cdots &
            z_1^{i_p} \\ \vdots & \ddots & \vdots \\  z_k^{i_1} & \cdots &
            z_k^{i_p} \end{array} }.
      \end{equation}
         \item[$\bullet$] For $\z=(z_1,\ldots,z_k) \in \C^k$ we define the
    $k\times k$ matrix $M_I\para{\z}$ by:
      \begin{equation} \label{MI}
      M_I\para{\z}=\displaystyle \left(\sum_{i \in I} \para{z_s \overline{z}_t}^i
      \right)_{s,t=1, \ldots, k}.\end{equation}
 Note that $M_I\para{\z}=  V_I\para{\z}^* V_I\para{\z}$.
    \item[$\bullet$] For $I\subset \N$ we define
$\mathcal{R}_{I} := \set{(z_1,\ldots,z_k) \in \C^k \;|\;
        \rank( V_{I}(\z))=k}$. For $I,J\subset \N$ we define $\mathcal{R}_{I,J}
        := \mathcal{R}_{I}\cap \mathcal{R}_{J}$.
     \item[$\bullet$] For $I,J\subset \N$ and $f,g\in \C[x]$ we define
     the set
{\small$$ \Omega_{I,J,k}(f,g):=\set{(\tilde{f}, \tilde{g})\,|\, \;
\exists (z_1, \ldots, z_k)\in\mathcal{R}_{I,J} \; \forall i \;
\tilde{f}(z_i)=\tilde{g}(z_i)=0;\; f-\tilde{f}\in \C[x]_I,g-
\tilde{g}\in \C[x]_J }.
$$}
Informally, $\Omega_{I,J,k}(f,g)$ is the set of pairs with at least $k$
common roots which are obtained from $(f,g)$ via  perturbation of the
coefficients corresponding to $I$ and $J$, respectively. We may omit
$(f,g)$ from $\Omega_{I,J,k}(f,g)$ if it is clear from the context.
  \end{itemize}
\end{definitions}

Next we introduce a family of polynomials which can be viewed as the
generalization of the Lagrange polynomials.
\begin{definition}
  Let $\z \in \mathcal{R}_I$ and let $V_I\para{\z}$ be the generalized
  Vandermonde matrix associated with $\z$ and $I$. Define $\x_I =
  \para{x^{i_1}, \ldots, x^{i_p}}$ and denote by $\set{\e_1,\ldots,\e_k}
  \subset \C^k$ the standard basis of $\C^k$. We define the {\bf
    generalized Lagrange polynomials} with support in $I$ as follows:
  \begin{equation} \label{defLIj}
    \tL_{I,i}(\z,x) := \x_I  V_I\para{\z}^\dagger  \e_i\;\;\;\; i =1,\ldots,k.
  \end{equation}
\end{definition}

Note that if $I=\set{0,\ldots,k-1}$ then $\set{L_{I,i}\para{\z,x}\,|\,1\leq i\leq k}$
are the classical Lagrange interpolation polynomials.

The following propositions assert that the generalized Lagrange
polynomials allow us to find the minimal norm  polynomial with
prescribed support $I$ satisfying (\ref{conditions1}). We also
highlight the connection between the 2-norms of the interpolation
polynomials and the results of Karmarkar and Lakshman in
\cite{KarLak1998}.

\begin{proposition}\label{prop1}
Let $I  \subset \N$ with $p \geq k$ and ${\bf z}=(z_1, \ldots, z_k) \in
\mathcal{R}_I$. Then for all $1\leq i,j\leq k, \,\tL_{I,i}({\bf z},z_j )=\delta_{i,j}.$
\end{proposition}

\begin{proof} From (\ref{defLIj}) we get that  $\tL_{I,i}(\z,z_j)=
{\bf e}_j^T V_I(\z) V_I(\z)^\dagger {\bf
    e}_i$ for all $i,j \in \set{1,\ldots,k}$. Then we use that $ V_I(\z)$
    has rank $k$ to get that  $V_I(\z)^\dagger$ is the right inverse of
    $ V_I(\z)$, thus $V_I(\z) V_I(\z)^\dagger=id$.
\end{proof}

\begin{proposition} \label{prop2}
   Let $I \subset \N$, $\z \in \mathcal{R}_I$ and $\mathbf{f}=(f_1,\ldots,f_k)\in \C^k$. Define
\begin{equation}\label{F(x)}
 F(x):=\displaystyle \sum_{i =0}^{k} f_i
   L_{I,i}(\z,x).
\end{equation}
 Then we have $F(x) \in \C[x]_I$ and
   \begin{equation} \label{conditions}
         F(z_j)=f_j, \forall j \in \set{1,\ldots,k}.
   \end{equation}
   Moreover,
   \begin{equation}\label{normform}
   \|F\|^2=\mathbf{f}^* M_I\para{\z}^{-1} \mathbf{f}
   \end{equation} is
   minimal among the polynomials in $\C[x]_I$ satisfying (\ref{conditions}).
\end{proposition}

\begin{proof}
 Let $F(x)$ be as in (\ref{F(x)}). If we denote by
 $\mathbf{F} = (F_i)_{i \in I}$ the vector of coefficients of   $F(x)$ then by the definition of the generalized Lagrange polynomials
   we have ${\bf F} = V_I({\bf z})^\dagger{\bf f}$.
 It is easy to check that
 $
 \|{\bf F}\|^2
 ={\bf f}^* M_I^{-1} {\bf f}
 $
 using the fact that $M_I^{-1}=V_I({\bf z})^{+*}V_I({\bf z})^\dagger$. On the other hand, ${\bf F}$ is the
 minimal 2-norm vector satisfying $V_I({\bf z}){\bf F}={\bf f}$,
 which follows from  the properties of the Moore-Penrose pseudo-inverse
 (see \cite{GoVa96}). Finally, we note that  $V_I({\bf z}){\bf F}={\bf f}$ is equivalent to  (\ref{conditions}).
\end{proof}

 The  above propositions allow us to state the main result of the
 subsection:
 \begin{theorem} \label{GenKaLa}
   Let $f,g \in \C[x]$, $I, J \subset
   \N$ and $\z \in \mathcal{R}_{I,J}$.  We define the following polynomials in
   $\C[x]_I$ and $\C[x]_J$, respectively:
   \begin{eqnarray} \label{FIGJ}
     F_I(\z,x) := \displaystyle \sum_{i=1}^k f(z_i) \tL_{I,i}(\z,x),\;
     G_J(\z,x) := \displaystyle \sum_{i=1}^k g(z_i) \tL_{J,i}(\z,x) .
   \end{eqnarray}
   Then
$$(f(x)-F_I(\z,x),g(x)-G_J(\z,x)) \in
   \Omega_{I,J,k}(f,g).$$
Moreover, if $\;\min_{\z \in \mathcal{R}_{I,J}}\para{\mathbf{f}^*
     \, M_I\para{\z}^{-1} \, \mathbf{f} + \mathbf{g}^* \,
     M_J\para{\z}^{-1} \, \mathbf{g}}
$ exists and is reached at
     $\zeta\in \mathcal{R}_{I,J}$ then we have
     $$\| F_I(\zeta,x)\|^2 + \| G_J(\zeta,x)\|^2 = \min_{(\tilde{f},\tilde{g}) \in
       \Omega_{I,J,k}} \set{\|f - \tilde{f}  \|^2 + \|g - \tilde{g}\|^2}.$$
       Here ${\bf f}=(f(z_1), \ldots, f(z_k))\in \C^k$ and ${\bf
g}=(g(z_1), \ldots, g(z_k))\in \C^k$.
 \end{theorem}

\begin{proof}
  The proof can be deduced easily from the proposition \ref{prop2}.
\end{proof}

\subsection{Generalized Weierstrass map}

In this section we give a generalization of the univariate
over-constrained Weierstrass map introduced in  \cite{ORthese}. 
 Informally, for $f,g\in \C[x]$ the Weierstrass map ${\mathcal W}$
 in \cite{ORthese}
  is a map
defined on $\C^k$ with the property that ${\mathcal W}(z_1, \ldots,
z_k)=0$ if and only if $f(z_i)=g(z_i)=0$ for $1\leq i\leq k$. 
 The main contribution of this subsection is the observation that the norm
 $\|{\mathcal W}({\bf z})\|_2$ is closely related to the distance
 defined by Karmarkar and Lakshman in \cite{KarLak1998}.
 Using this observation, it is
 straightforward to see  that the least square minimum of the Weierstrass map ${\mathcal W}$  corresponds to the $k$ common roots  of the
 closest system $\tilde{f},\tilde{g}$ which is obtained
 from $f,g$ via the perturbation of a prescribed subset of their
 coefficients.\\

 First we give the  definition of the generalized  Weierstrass map
 using the generalized Lagrange polynomials defined in (\ref{defLIj}).

 \begin{definition} \label{Weierstrassmap}
  Let $f,g\in \C[x]$, $k\geq 1$ and $I,J\subset \N$ such that
 $|I|,|J|\geq k$.
 For a fixed ${\bf z}\in {\mathcal R}_{I,J}$, let $F_I({\bf z}, x)\in \C[x]_{I}$ and
 $G_J({\bf z}, x)\in \C[x]_{J}$ be the interpolation polynomials
 defined in (\ref{FIGJ}).
 Then  the map defined by
\begin{eqnarray}\label{Wmap}
{\mathcal W}_{I,J}: \;\;\begin{cases} \C^k\rightarrow \C[x]_I\oplus
\C[x]_J\\
 {\bf z}\mapsto \left(F_I({\bf z}, x), G_J({\bf z}, x)\right)
 \end{cases}
 \end{eqnarray}
 is called the {\bf generalized Weierstrass map} with supports $I$ and
 $J$.
 \end{definition}

In the next theorem we prove that the least square solution  of the
Weierstrass map and the optimization problem posed by Karmarkar and
Lakshman in \cite{KarLak1998} are closely related.

\begin{theorem}\label{WmapProp}
 Let ${\bf z}=(z_1, \ldots, z_k)$, $(f,g)$, and ${\mathcal
W}_{I,J}$  be  as in Definition \ref{Weierstrassmap}. Then
\begin{enumerate}
\item ${\mathcal W}_{I,J}({\bf z})=0$ if and only if $(z_1, \ldots,
 z_k)$ are common roots of $f$ and $g$.

 \item Using the notation
 of Theorem \ref{GenKaLa},
 for all ${\bf z}\in \C^k$ we have
 \begin{eqnarray*}
 \|{\mathcal W}_{I,J}({\bf z})\|^2=
  {\bf f }^*M_I^{-1} {\bf f} + {\bf g}^* M_J^{-1} {\bf g}.
  \end{eqnarray*}

 \item
 $\min_{{\bf z}\in \mathcal{R}_{I,J}}\|{\mathcal W}_{I,J}({\bf z})\|^2 =
 \min_{(\tilde{f},\tilde{g}) \in \Omega_{I,J,k}}\set{\| f - \tilde{f}
   \|^2 + \| g - \tilde{g} \|^2}$.
 \end{enumerate}
  \end{theorem}

\begin{proof}
(i) ${\mathcal W}_{I,J}({\bf z})=0$ if and only if  $F_I({\bf z},
x)=G_I({\bf z}, x)=0$ for all $x\in \C$. This implies that
 $f(z_i)=F_I({\bf z},z_i )=0$ and $g(z_i)=G_I({\bf z},z_i )=0$ for all
 $1\leq i\leq k$. On the other hand, assume that  $z_1, \ldots z_k$
 are common roots of $f$ and $g$.
 Since $F_I$ and $G_J$ are the minimal 2-norm polynomials interpolating
 $(f(z_1), \ldots, f(z_k))=0$ and  $(g(z_1), \ldots, g(z_k))=0$,
  $F_I$ and $G_J$ must both be  the zero polynomial.\\
(ii) follows from the definition of ${\mathcal W}_{I,J}$ in
(\ref{Wmap}), the
 definition of $F_I({\bf z},
x)$ and $G_I({\bf z}, x)$ in (\ref{FIGJ}) and from (\ref{normform}).\\
(iii) follows from (ii) and from Theorem \ref{GenKaLa}.
\end{proof}

\begin{remark}\label{remark2}
As a special case of the above proposition, we get that the least
squares solution of the univariate over-constrained  Weierstrass map
${\mathcal W}$ defined in \cite{ORthese} gives the common roots of the
closest system with $k$ common roots, and obtained via the perturbation
of the coefficients corresponding to $I=J=\{0, 1, \ldots, k-1\}$, i.e.
the terms of $f$ and $g$ of degree less than $k$. This gives a link between the Weierstrass map of  \cite{ORthese} and the 
distance formulated for the approximate GCD problem by Karmarkar and Lakshman in  \cite{KarLak1998}.
\end{remark}

\subsection{Gauss-Newton iteration}

 In Theorem \ref{WmapProp} we obtained a formulation for the distance of $f,g$ from the set of pairs with at least $k$ common roots as the 2-norm minimum of the Weierstrass map   ${\mathcal W}_{I,J}$. In this subsection we give explicit formulas for the Gauss-Newton iteration for  ${\mathcal W}_{I,J}$. The theoretical framework for the Gauss-Newton iteration for computing the 2-norm optimum of complex  functions is described in the multivariate setting in Section \ref{methods}, in the present subsection  we present our results without proof.

 First we would like to note that    if $|I|\neq k$ or $|J|\neq k$  then the function ${\mathcal
W}_{I,J}({\bf
 z})$ is not a complex analytic function. However, in this case we can separate the original complex variables ${\bf z}=(z_1, \ldots, z_k)\in \C^k$ and their conjugate $\bar{\bf z}=(\bar{z}_1, \ldots, \bar{z}_k)\in \C^k$, and express ${\mathcal W}_{I,J}$ as a function of both of them. A simple computation described in Section \ref{methods} shows that the vanishing of the gradient of  $\|{\mathcal W}_{I,J}\|^2$ will result in two equations which are conjugates of each other. Thus  solving only one of them will result to the definition of the Gauss-Newton iteration as follows (see more details in Section \ref{methods}):
 \begin{eqnarray}\label{GaussNewton}
{\bf z}^{new}= {\bf z}-{\mathcal J}({\bf z})^\dagger\  {\mathcal
W_{I,J}}({\bf z}),
 \end{eqnarray}
  where  ${\mathcal J}(\bz)$
 is  the Jacobian matrix of
  ${\mathcal W}_{I,J}$  at ${\bf z}$ of size $(|I|+|J|)\times k$. \\

In the following proposition we give an expression of the Gauss-Newton iteration computed by conducting linear algebra on the Vandermonde matrices $V_I(\bz)$ and $V_J(\bz)$. 
 
\begin{proposition}  Let $f,g\in
\C[x]$, $k> 0$, and $I, J\subset \N$ such that $|I|,|J|\geq k$. For a fixed
k-tuple $\bz =(z_1, \ldots, z_k)\in \C^k$ of distinct numbers define  
\begin{eqnarray}\label{fzgz}
{f}_{\bz}(x):=f-F_I(\bz,x) \text{ and } {g}_{\bz}:=g-G_J(\bz,x)
\end{eqnarray} using  (\ref{FIGJ}).  The iteration defined by
 \begin{eqnarray}\label{Weierit2}
 {\bz^{new}}:=\bz - \left( D_{f_\bz}^* M_I^{-1} D_{f_\bz} +
 D_{g_\bz}^* M_J^{-1} D_{g_\bz}
 \right)^{-1} \left( D_{f_\bz}^* M_I^{-1}{\bf f} +D_{g_\bz}^* M_J^{-1} {\bf
 g}\right)
 \end{eqnarray}
 is the Gauss-Newton  iteration  for the Weierstrass map ${\mathcal W}_{IJ}$.  Here
 \begin{eqnarray*}
 D_{f_\bz}=\diag\left( f'_\bz (z_i)\right)_{ i=1}^k,\;\;
D_{g_\bz}=\diag\left( g'_\bz (z_i)\right)_{ i=1}^k\in \C^{k\times k}.
\end{eqnarray*}
\end{proposition}

\begin{proof}
This is a special case of the formula (\ref{MultIt2}) described in Section \ref{methods}.
\end{proof}

\subsection{Simplified iteration} 

The simplification we
propose is analogous to the idea used in the classical univariate Weierstrass iteration, which we briefly
describe first. The classical univariate Weierstrass iteration finds simultaneously all roots of a given monic
univariate polynomial $f$ of degree $k$, and  has the following simple and elegant component-wise iteration
function:
 $$ z_i^{new} = z_i- \frac{f(z_i)}{\prod_{j\neq i}(z_i-z_j)}  \quad i=1, \ldots, k.
 $$
 One can derive this formula by applying the
Newton method to the corresponding Weierstrass map, and then expressing the result in terms of the standard
Lagrange polynomial basis at the iteration point: the Jacobian of the Weierstrass map is diagonal  in the
Lagrange basis, which results in the simple, component-wise iteration formula. Generalization of this  to
finding the roots of multivariate systems were proposed in  \cite{OR01}.

Now we explore an analogue of the above simplification to our problem of solving approximate over-constrained systems. 
First we need to make the following assumption about the size of the support of the perturbation functions: 

\begin{center}{\bf Assumption:}  $|I|=|J|=k$.\end{center}

We will need the following lemma:

 \begin{lemma}\label{derivative} Let $f$, ${\bf z}$, $I$,  and
 $F_I({\bf z},x)$ be  as in
 Definition \ref{Weierstrassmap} and assume that $|I|=k$. Let $\tL_{I,1}, \ldots,
 \tL_{I,k}$ be the Lagrange polynomials defined  in (\ref{defLIj}).  Then for all $1\leq i\leq k$ we have
 \begin{eqnarray}\label{Derivate}
    \frac{\partial F_I({\bf z},x)}{\partial z_i}=
  \left(f'(z_i)-F'_I({\bf z},z_i) \right)\tL_{I,i}({\bf z},x).
 \end{eqnarray}
\end{lemma}

\begin{proof}
Implicitly differentiating the equations
$$
F_I({\bf z},z_j)=f(z_j) \;\; j=1, \ldots, k
$$ by $z_i$ we get
$$
\frac{\partial F_I({\bf z},x)}{\partial z_i}\left|_{x=z_j}\right.
+\delta_{i,j} \frac{\partial F_I({\bf z},x)}{\partial
x}\left|_{x=z_j}\right.= \delta_{i,j}\frac{\partial f(x)}{\partial
x}\left|_{x=z_j}\right..
$$
By the assumption that $|I|=k$ we have that
$$\langle \tL_{I,1}, \ldots,
 \tL_{I,k}\rangle=\C[x]_I,$$ which implies that $\frac{\partial F_I({\bf z},x)}{\partial
z_i}$  is equal to the expression in the claim.
\end{proof}


 \begin{definition}\label{DefWeieriteration}
 Let $(f,g)$, $k$, $\bz=(z_1, \ldots, z_k)$, $I$, $J$, $F_I({\bf
 z},x)$, and $G_J({\bf
 z},x)$ be as in Definition \ref{Weierstrassmap}. Assume that
 $|I|=|J|=k$.
 As in (\ref{fzgz}),  let
 \begin{eqnarray*}
 {f}_{\bz}(x):=f(x)- F_I({\bf z},x), \;\; {g}_{\bz}(x):=
 g(x)- G_J({\bf
 z},x).\end{eqnarray*}
 Assume that none of the $z_i$'s are common roots of the derivatives
 ${f}'_{\bz}(x)$ and
 ${g}'_{\bz}(x)$.
 Then   the {\bf
 simplified Gauss-Newton iteration} with supports $I$ and $J$ is defined by
 \begin{eqnarray}\label{Weieriteration}
 z'_i:= z_i-\frac{\overline{{f}'_{\bz}(z_i)}f(z_i)+
 \overline{{g}'_{\bz}(z_i)}g(z_i)}{|{f}'_{\bz}(z_i)|^2+|{g}'_{\bz}(z_i)
 |^2}  \;\;i=1,\ldots, k.
 \end{eqnarray}

\end{definition}

Note that (\ref{Weierit2}) equals
(\ref{Weieriteration}) if we replace $M_I$ and $M_J$ by the identity matrix in (\ref{Weierit2}) and exploit the diagonality of the matrices $D_{f_\bz}$ and $D_{g_\bz}$ to obtain the component-wise formulation.\\

The following theorem asserts that $\bz\in \C^k$ are fixed points of the simplified  Gauss-Newton iteration
if  the corresponding perturbation functions are pointwise minimal in a neighborhood of $\bz$. 

\begin{theorem}\label{UniSimpTheorem}
A point  $\bz =(z_1, \ldots, z_k)\in \C^k$ is a fixed point of the simplified Gauss-Newton iteration defined in (\ref{Weieriteration})
 if  
 there
exists an  open neighborhood $U$ of $\bz$  such
 that for all $\tilde{\bz}=(\tilde{z}_1, \ldots, \tilde{z}_k)$ and $\bz'=(z'_1,\ldots, z'_k)$ in $U$ 
  \begin{eqnarray}\label{pointwise}
  \left|F_I(\bz, \tilde{z_i})\right|^2+\left|G_J(\bz, \tilde{z_i})\right|^2\le \left|F_I(\bz', \tilde{z_i})\right|^2+\left|G_J(\bz', \tilde{z_i})\right|^2 \quad i=1,\ldots, k.
  \end{eqnarray}
  Note that this includes the case when $z_1, \ldots, z_k$ are common roots of $f$ and $g$, in which case $F_I(\bz, x)=G_J(\bz,x)=0$. 
\end{theorem}

\begin{proof}
Assume $\bz\in \C^k$ satisfies the condition in (\ref{pointwise}) for some neighborhood $U$. 
Then for any fixed $\tilde{\bz}=(\tilde{z}_1, \ldots, \tilde{z}_k)\in U$ we have that for all $i=1, \ldots ,k$
$$
\frac{\partial}{\partial z_i}\left(  \left|F_I(\bz, \tilde{z_i})\right|^2+\left|G_J(\bz, \tilde{z_i})\right|^2\right)=0.
$$
Thus, 
$$\frac{\partial F_I(\bz, \tilde{z_i})}{\partial z_i} \overline{F_I(\bz, \tilde{z_i} )}+ 
\frac{\partial  G_J(\bz, \tilde{z_i})}{\partial z_i}\overline{G_J(\bz, \tilde{z_i} )} =0.
$$
Using Lemma \ref{derivative} we get that
$$
  \left(f'(z_i)-F'_I({\bf z},z_i) \right)\tL_{I,i}({\bf z},\tilde{z_i})  \overline{F_I(\bz, \tilde{z_i} )} +
  \left(g'(z_i)-G'_J({\bf z},z_i) \right)\tL_{I,i}({\bf z},\tilde{z_i})  \overline{G_J(\bz, \tilde{z_i} )}=0.
  $$
  In particular, as  $\tilde\bz$ approaches $\bz$ we get that
  $$
  \left(f'(z_i)-F'_I({\bf z},z_i) \right)  \overline{f(z_i)} +
  \left(g'(z_i)-G'_J({\bf z},z_i) \right) \overline{g(z_i)}=0.
 $$
 Using the definition of $f_\bz$ and $g_\bz$ we get that
 $f'_\bz(z_i)\overline{f(z_i)} +g'_\bz(z_i)  \overline{g(z_i)}=0
 $, and the left hand side is the conjugate of the numerator of the iteration function  in
(\ref{Weieriteration}). This proves the claim.

\end{proof}
\section{Multivariate Case}

In this section, we describe the generalization of the results of the
previous section to the multivariate setting. In the multivariate case
we extend our construction to over-constrained  systems of analytic
functions as input, not only polynomials. Since the set of
over-constrained systems of analytic functions with at least $k$ common
roots is infinite dimensional, we will restrict our objective to find
the closest such system which is obtained via some perturbation from a
finite dimensional  ``perturbation space", given by a finite basis of
analytic functions. In order to handle analytic functions as input, we
assume that they are given in a ``black box" format, i.e. we assume
that we can evaluate these functions in some  fixed precision in unit time
 at any point. For our general construction we
need to generalize the Lagrange interpolation to finding elements in
the perturbation space with prescribed evaluations and minimal 2-norms.

\begin{definition}\label{firstanal}
We denote by $\C^\infty_n$ the set of analytic functions
$\C^n\!\rightarrow \!\C$.  Let $\vec{f}=(f_1, \ldots,
f_N)\in(\C^\infty_n)^N$ for some $N>n$. For each $i=1, \ldots, N$ let $
B_i:=\{b_{i,1}, \ldots, ,\ldots,b_{i,m_i}\}\subset\C^\infty_n$
 linearly independent over $\C$.
  We call ${\mathcal P}:=
 \bigoplus_{i=1}^N{\rm span}_\C(B_i)$ the {\bf perturbation space} with basis
 $\vec{B}:=(B_1, \ldots, B_N)$.
 \end{definition}

We address
the following problem:\\

 \noindent{ {\bf Problem}: {\it Given $\vec{f}=(f_1, \ldots, f_N)$ and
 $\vec{B}=(B_1, \ldots, B_N)$ as above. Find $(p_1, \ldots, p_N)\in {\mathcal P}$
 such that $(f_1-p_1, \ldots, f_N-p_N)$ has at least $k$ distinct common
 roots in
 $\C^n$ and $\|p_1\|_{B_1}^2+\cdots \|p_N\|_{B_N}^2$ is minimal.
 Here $\|p_i\|_{B_i}$ denotes the 2-norm  of the coefficients of $p_i$
 in the $\C$-basis $B_i$. }\\

Let us define the generalized Vandermonde matrix associated with a
 set of basis functions $B$:
 \begin{definition}\label{VB}
 Let $\vz=(\bz_1, \ldots, \bz_k)\in (\C^n)^k$. For  $B=\{b_1,
 \ldots , b_m\}\subset \C^\infty_n$
 we define the {\bf generalized Vandermonde matrix} associated with $B$
 to be
 the $k\times m$ matrix with entries
 $$
 V_B(\vz)_{i,j}:=b_j(\bz_i).
 $$
 We denote
 $${\mathcal R}_{B} :=
 \set{\vz \in \para{\C^n}^k \, | \, \rank \para{V_B\para{\vz}} = k
 \,},$$
 and for $\vec{B}=(B_1, \ldots, B_N)$ we we use the notation ${\mathcal
 R}_{\vec{B}}:=\bigcap_{i=1}^N {\mathcal R}_{B_i}$.
\end{definition}

 \begin{remark} We can choose the bases $B_1, \ldots, B_N$ of
 the perturbation space
 freely as
 long as the set ${\mathcal
 R}_{\vec{B}}$ is open and everywhere dense,
 or it includes the possible roots we are searching for.
\end{remark}

  Now we can define the generalized multivariate Lagrange polynomials :

 \begin{definition} \label{MLagrange}
 Let $B=\{b_1, \ldots, b_m\}\subset \C^\infty_n$.
  For $\bx \in \C^n$ denote $\bx_B =[b_1(\bx),
 \ldots,b_m(\bx)]$. Let ${\bf e_1}\ldots{\bf e_k}$ be the standard basis of $\C^k$.
 Let $\vz \in {\mathcal R}_{B}$. We define the {\bf generalized
 Lagrange polynomials}
   associated with
  $B$ as $\tL_{B,i}(\vz,\bx):=\bx_B V_B(\vz)^\dagger {\bf e_i}$ for  $i =1,\ldots,k$.
 \end{definition}

 \begin{remark} If $m=k$ and $B=\{\bx^{\alpha_1}, \ldots, \bx^{\alpha_k}\} $
 for some $\alpha_i\in \N^n$, then the generalized Vandermonde matrix is
  a  square matrix and the above formula is
 the one given by Ruatta in \cite{ORthese} for the Lagrange interpolation basis.
 \end{remark}

  The following proposition is a straightforward generalization of
  Propositions \ref{prop1} and \ref{prop2}.

 \begin{proposition} \label{MKLProp}
 Let $f \in \C^\infty_n$, $B\subset \C^\infty_n$, $|B|=m$
 linearly independent over $\C$, and let ${\mathcal P}={\rm span}_\C(B)$.
  Fix $\vz=(\z_1, \ldots, \z_k) \in {\mathcal R}_B$. Then
 $\tL_{B,i}\para{\vz,\z_j} = \delta_{i,j}$   for all
 $i,j =1,\ldots,k$.
  Furthermore,
  define
   $p(\vz, \bx) = \displaystyle
 \sum_{i=1}^k f(\z_i) \tL_{B,i}\para{\vz,\bx}\in {\mathcal P}.$ Then
 \begin{equation} \label{Mconditions}
 p(\vz, \z_j)=f(\z_j) \mbox{ for all } j \in \set{1,\ldots,k}.
 \end{equation}
   Moreover,
   \begin{equation}\label{Mnormform}
   \|p\|_B^2=\mathbf{f}^* M_B\para{\vz}^{-1} \mathbf{f}
   \end{equation} is
   minimal among the polynomials in ${\mathcal P}$ satisfying (\ref{Mconditions}).
   Here
  \begin{equation} \label{ME}
  {\bf f} := \left(f(\z_1), \ldots, f(\z_k)\right)^T \;\text{ and }\; M_B\para{\vz} := V_B\para{\vz} V_B\para{\vz}^* =
    \para{\displaystyle \sum_{b \in B}
    b(\z_i) \overline{b(\z_j)}}_{i,j \in \set{1,\ldots , k}}.
  \end{equation}
 \end{proposition}

 The next theorem gives a generalization of the expressions of Karmarkar
and Lakshman in \cite{KarLak1998} for the multivariate case. This is
one of the main results of the paper.

 \begin{theorem}\label{MLagTheorem}
 Let $N>n \in \N$, $\vec{f}=(f_1,\ldots,f_N) \in (\C^\infty_n)^N$,
 $\vec{B} = \para{B_1,\ldots,B_N}$ and ${\mathcal P}=\bigoplus_{i=1}^N{\rm span}_\C(B_i)$ be as in Definition \ref{firstanal}.
 Define  $\mathbf{f}_i(\vz) :=\para{f_i(\z_1),\ldots,f_i(\z_k)}\in
 \C^k$ and let $M_{B_i}(\vz)$ be as in (\ref{ME})
 for $i=1, \ldots, N$. Then,  if
 \begin{equation}\label{minz}\min_{\vz \in {\mathcal R}_{\vec{B}}}
 \mathbf{f}_1^* M_{B_1}^{-1} \mathbf{f}_1 \para{\vz}+ \cdots +
 \mathbf{f}_N^* M_{B_N}^{-1}
 \mathbf{f}_N\para{\vz} \end{equation} exists,  it is equal to
 \begin{equation}\label{minOmega}\min_{\tilde{f} \in \Omega_{\vec{B},k}(\vec{f})}
\|{f}_1 - \tilde{f}_1\|_{B_1}^2 + \ldots +\|{f}_N - \tilde{f}_N
 \|_{B_N}^2.\end{equation}
 Here the minimum is taken within the set  $\Omega_{\vec{B},k}({f})$ defined by
  \begin{equation*}
 \Omega_{\vec{B},k}(\vec{f}):=\set{\tilde{f}=(\tilde{f}_1,\ldots,\tilde{f}_N)
 \,:\,
   \forall i \; f_i-\tilde{f}_i \in {\rm span}_\C(B_i), \; \exists
   (\z_1, \ldots, \z_k)\in {\mathcal R}_{\vec{B}}\; \forall i,j \;
   \tilde{f}_i(\z_j)=0}.
 \end{equation*}
 \end{theorem}

 \begin{proof}
 For a fixed $\vz \in {\mathcal R}_{\vec{B}}$
 define $p_{i}(\vz,\bx) := \sum_{j =1}^k f_i(\z_j)
 \tL_{B_i,j}(\vz,\bx)\in  {\rm span}_\C(B_i)$ for all
 $i =1,\ldots,N$.
 Assume that the minimum in (\ref{minz}) exists and is taken at
 $\vec{\zeta}=(\zeta_1, \ldots, \zeta_k) \in {\mathcal R}_{\vec{E}}$.
   Note that   for
all $i \in \{1,\ldots,N\}$, if  $\tilde{f}_i$ vanishes on  $\zeta_1, \ldots,
\zeta_k$ and $f_i-\tilde{f}_i\in {\rm span}_\C(B_i)$, then, by
Proposition \ref{MKLProp},
 $ \|f_i-\tilde{f}_i\|_{B_i}\geq \|p_{i}(\vec{\zeta}, \bx)\|_{B_i} $. This implies that
$$
\left(f_1(\bx)-p_{1}(\vec{\zeta},
\bx),\ldots,f_N(\bx)-p_{N}(\vec{\zeta}, \bx)\right)\in
\Omega_{\vec{B},k}(\vec{f})
$$ must minimize (\ref{minOmega}). The equality of (\ref{minz}) and
(\ref{minOmega}) follows from
 \begin{equation*}
 \|{p}_{1}(\vec{\zeta}, \bx)\|_{B_1}^2 + \cdots +
 \|{p}_{N}(\vec{\zeta}, \bx) \|_{B_N}^2 =
 \mathbf{f}_1^* M_{B_1}^{-1} \mathbf{f}_1(\vec{\zeta}) + \cdots +
 \mathbf{f}_N^* M_{B_N}^{-1}
  \mathbf{f}_N(\vec{\zeta}).
 \end{equation*}
 \end{proof}

 Next we define the multivariate generalization of the Weierstrass map :
 \begin{definition}\label{multWeiermap}
 Let $f_1,\ldots,f_N \in \C^\infty_n$, $\vec{B}=\para{B_1,\ldots,B_N}$
 and ${\mathcal P}$ be
 as above.
 The generalized Weierstrass map is defined as follows:
 \begin{equation}
   \mathcal{W}_{\vec{B}} : \left\{ \begin{array}{ccc}
   {\mathcal R}_{\vec{B}} & \longrightarrow & \displaystyle
   {\mathcal P} \\
    \vec{\z} & \longmapsto & \para{ \begin{array}{c}
    p_{1}\para{\vec{\z},\bx} \\ \vdots \\  p_{N}\para{\vec{\z},\bx}
    \end{array} } \end{array} \right.  ,
 \end{equation}
 where
 $$p_{i}(\vz,\bx) := \sum_{j =1}^k f_i(\z_j)
 \tL_{B_i,j}(\vz,\bx) \;\;\;
 i =1,\ldots,N.
 $$
 \end{definition}

 The next proposition is a straightforward generalization of  Proposition \ref{WmapProp} :

 \begin{proposition}
 Let $\vec{f}=(f_1,\ldots,f_N )\in (\C^\infty_n)^N$,
 $\vec{B}=\para{B_1,\ldots,B_N}$
  be as above.
 Then for all $\vec{\z} \in  \mathcal{R}_{\vec{B}}$
 we have $\mathcal{W}_{\vec{B}}\para{\vec{\z}}=0$ if and only if
 $\set{\z_1,\ldots,\z_k}$ are common roots of $f_1,\ldots,f_N$.
 Moreover, using the notation of Theorem \ref{MLagTheorem}, we have
 \begin{equation}
 \min_{\vec{\z} \in {\mathcal R}_{\vec{B}}}
 \| \mathcal{W}_{\vec{B}}\para{\vec{\z}} \|^2=
 \min_{\tilde{f} \in \Omega_{\vec{B},k}(\vec{f})}
 \|{f}_1 - \tilde{f}_1\|_{B_1}^2 + \ldots +\|{f}_N - \tilde{f}_N
 \|_{B_N}^2.\end{equation}
 \end{proposition}

In the rest of the paper we will describe iterative methods to
approximate the minimum $$\min_{\vec{\z} \in {\mathcal R}_{\vec{B}}}
 \| \mathcal{W}_{\vec{B}}\para{\vec{\z}} \|^2.$$

\section{Numerical methods}\label{methods}

In this section we describe the iterative methods we use in our numerical experiments for comparison. These methods try to minimize the squared 2-norm of a function $W:U\rightarrow \C^T$ for  some open subset $U\subseteq \C^S$, by  approximating it by its truncated Taylor series expansion. 

\subsection{Gauss-Newton method}

Using the previous notation, in our case $W:= {\mathcal W}_{\vec{B}}: {\mathcal R}_{\vec{B}}\rightarrow {\mathcal P}$ is the generalized Weierstrass map defined in Definition \ref{multWeiermap}, such that its image is expressed as the vector of coefficients of the perturbation polynomials in ${\mathcal P}$.  
 We denote by $\nabla$  the vector of derivations by the variables $\vec{\bf z}$ (and not by their conjugates), and $J=\nabla W$. We also denote by $\overline{\nabla}$  the vector of derivations by the conjugate variables. To minimize indicies and simplify the notation, we denote by $z_i$ and $\overline{z}_i$ the coordinates of $\vec{\bf z}$ and their conjugates.

First we argue that it is sufficient to consider only derivations by the variables $\vec{\bf z}$ and not by their conjugates when we define the Gauss-Newton method. We need the following lemmas: 

\begin{lemma}
 Let $\vec{f}=(f_1, \ldots, f_N)$, $\vec{B}=(B_1, \ldots, B_N)$, $\vecz\in {\mathcal R}_{\vec{B}}\subset(\C^n)^k$ as in Definition \ref{multWeiermap}. Define $F$  to be the column vector
$$
F:=\left(f_1(\bz_1), \ldots, f_1(\bz_k), \ldots, f_N(\bz_1), \ldots, f_N(\bz_k)\right)^T\in \C^{kN}
$$
and the matrix
$
\FullVandermonde$ as the block diagonal matrix of size $(kN)\times \left(\sum |B_i|\right)$, with diagonal blocks the Vandermonde matrices $V_{B_1}(\vecz), \ldots,V_{B_N}(\vecz)$.  Then the gradient of the Weirstrass map $W$ is
\begin{eqnarray}\label{JacobianOfC}
J=\nabla \CoeffVec= \FullVandermonde^{\dagger} \left(\nabla F - (\nabla V) W\right).
\end{eqnarray}
\end{lemma}
\begin{proof}
By definition, $\CoeffVec$ is the least square solution of 
$$
F = \FullVandermonde \CoeffVec.
$$
The use of the Moore-Penrose pseudoinverse of $V$  can be described in two steps.  First we find
$G$ such that
\begin{equation} \label{E:g}
\FullVandermonde \FullVandermonde^*G = F
\end{equation}
then we compute $\CoeffVec$ as
\begin{equation}\label{E:Fg}
\CoeffVec = \FullVandermonde^*G.
\end{equation}
From equation (\ref{E:Fg}) we have
\begin{equation}\label{E:dF}
\nabla \CoeffVec = (\nabla \FullVandermonde^*)G +
        \FullVandermonde^*(\nabla G).
\end{equation}
From equation (\ref{E:g}) we have
\begin{equation} \label{E:dg}
\nabla G = \left(\FullVandermonde \FullVandermonde^*\right)^{-1}\left( \nabla F - (\nabla \FullVandermonde )\FullVandermonde^*G -
\FullVandermonde (\nabla \FullVandermonde^*)G\right) . 
\end{equation}
Combining equations (\ref{E:dF}) and (\ref{E:dg}) and using that fact that $\nabla \FullVandermonde^*=0$ gives
 $$
 \nabla W = \FullVandermonde^*\left(\FullVandermonde \FullVandermonde^*\right)^{-1}\left( \nabla F - (\nabla \FullVandermonde )\FullVandermonde^*(\FullVandermonde \FullVandermonde^*)^{-1}F\right) = \FullVandermonde^{\dagger}\left( \nabla F - (\nabla \FullVandermonde )\CoeffVec\right) . 
 $$
\end{proof}

 \begin{lemma}
$$
\frac{\partial W^*}{\partial z_i} V^\dagger =0.
$$
\end{lemma}

\begin{proof}
By definition $W = V^\dagger F = V^*(VV^*)^{-1}F$ and thus 
\begin{eqnarray*}
\frac{\partial W^*}{\partial z_i}
  = - F^* (VV^*)^{-1} \frac{\partial V}{\partial z_i} V^* (VV^*)^{-1} V
    + F^* (VV^*)^{-1} \frac{\partial V}{\partial z_i}.\end{eqnarray*}
Then we have 
\begin{eqnarray*}
\frac{\partial W^*}{\partial z_i}V^{\dagger}
 &=& F^* (VV^*)^{-1}
   \left(   \frac{\partial V}{\partial z_i} 
          - \frac{\partial V}{\partial z_i} V^* (V^\dagger)^* \right)
   V^\dagger \\
 &=& F^* (VV^*)^{-1}
   \left(   \frac{\partial V}{\partial z_i} V^\dagger 
          - \frac{\partial V}{\partial z_i} V^* (V^\dagger)^* V^\dagger
   \right) \\
 &=& F^* (VV^*)^{-1}
   \left(   \frac{\partial V}{\partial z_i} V^\dagger 
          - \frac{\partial V}{\partial z_i} V^* (VV^*)^{-1}V V^*(VV^*)^{-1}
   \right) \\
 &=& F^* (VV^*)^{-1}
   \left(   \frac{\partial V}{\partial z_i} V^\dagger 
          - \frac{\partial V}{\partial z_i} V^\dagger
   \right) \
 = 0 \quad \QED
\end{eqnarray*}
\end{proof}

\begin{corollary}
$$
\frac{\partial W^*}{\partial z_i} W =0, \quad  \frac{\partial W^*}{\partial z_i} J=0, \quad \text{and} \quad\frac{\partial \|W\|^2}{\partial z_i} = W^*\frac{\partial W}{\partial z_i} 
$$
\end{corollary} 
\begin{proof}
Follows from $W=V^\dagger F$ and $J= V^\dagger(\nabla F-(\nabla V)W)$. 
\end{proof}

\begin{corollary}
If we assume that $J^{*}(\xi)W({\xi})=0$ then 
$$\frac{\partial J^\dagger W}{\partial \overline{z}_i} ({\xi})= \left((J^*J)^{-1} \left(\frac{\partial J}{\partial {z}_i}\right)^* W\right) ({\xi}) .
$$
\end{corollary} 
\begin{proof}
Using that $J^{*}(\xi)W({\xi})=0$ we get
\begin{eqnarray*}
\frac{\partial J^\dagger W}{\partial \overline{z}_i} ({\xi})& =&\frac{\partial (J^*J)^{-1}J^* W}{\partial \overline{z}_i} ({\xi})\\
&=&\left((J^*J)^{-1} \frac{\partial J^*}{\partial \overline{z}_i} W\right) ({\xi})+ \left((J^*J)^{-1} J^*\frac{\partial W}{\partial \overline{z}_i}\right) ({\xi}) \\
&=&  \left((J^*J)^{-1} \left(\frac{\partial J}{\partial {z}_i}\right)^* W\right) ({\xi})  + \left((J^*J)^{-1}\left( \frac{\partial W^*}{\partial {z}_i}J\right)^*\right) ({\xi}) \
\end{eqnarray*}
and the last term is $0$ by the previous corollary. 
\end{proof}

\bigskip
The following argument is from \cite[Theorem 4]{DeSh2000}:

\begin{proposition}
Define the Gauss-Newton method by the map
\begin{eqnarray}\label{GaussNewtonIteration}
N_W(\vec{\bf z}):=\vec{\bf z}-J^{\dagger}(\vec{\bf z})W(\vec{\bf z}). 
\end{eqnarray}
Let $\xi\in {\mathcal R}_{\vec{B}}$ such that $J$ has full rank at $\xi$, 
$$
J^{*}(\xi)W({\xi})=0,
$$
and 
we have the following inequality:
\begin{equation}\label{ineq}
\|J^{\dagger}({\xi})\|^2\cdot \|\left[
\begin{array}{c|c}
\nabla J(\xi)&\nabla J^* (\xi)
\end{array}
\right] \|\cdot \|W(\xi)\|<1,
\end{equation}
where for a matrix $M$,  $\|M\|$ denotes the operator $2$-norm, i.e. $\|M\|={\sup}_{\|x\|=1} \|Mx\|, $ while for a 3-dimensional matrix $N$ it is 
 $\|N\|={\sup}_{\|x\|=1} \|N(x,x)\| $.
 Then  $\xi$ is an attractive fixed point for $N_W$.
\end{proposition}

\begin{proof}
To prove the claim we have that
$$ N_W(\xi, \overline{\xi})-N_W({\bf z}, \overline{\bf z})= \left[
\begin{array}{c|c}
\nabla N_W(\xi)& \overline{\nabla} N_W(\xi) 
\end{array}
\right] \cdot \left[
\begin{array}{c}
\xi-{\bf z}\\
\overline{\xi}-\overline{\bf z}
\end{array}
\right]
 +h.o.t.$$
Using that $J^{*}(\xi)W({\xi})=0$ we get that
$$
\nabla N_W(\xi)= \left(-(J^*J)^{-1}(\nabla J^*) W\right)(\xi),
$$
and also using the previous Corollary we have that
$$
 \overline{\nabla} N_W(\xi) =  \left(-(J^*J)^{-1}(\nabla J)^* W\right)(\xi).
 $$
 Therefore, 
$$\left[
\begin{array}{c|c}
\nabla N_W(\xi)& \overline{\nabla} N_W(\xi) 
\end{array}
\right] =-(J^*J)^{-1} (\xi)\left[
\begin{array}{c|c}
\nabla J(\xi)& \nabla J^*(\xi) 
\end{array}
\right] W(\xi), 
$$
and its norm is bounded by $
\|J^{\dagger}({\xi})\|^2\cdot \|\left[
\begin{array}{c|c}
\nabla J(\xi)&\nabla J^* (\xi)
\end{array}
\right] \|\cdot \|W(\xi)\|<1
$, which proves that $\xi$ is an attractive fixed point of $N_W$.
\end{proof}

Next we give an explicit formula for the Gauss-Newton iteration defined in (\ref{GaussNewtonIteration}) in terms of  $M_{B_i}(\vec{\bf z})=V_{B_i}(\vec{\bf z})V^*_{B_i}(\vec{\bf z})$ and the function values ${\bf f}_i(\vec{\bf z})$. 

\begin{proposition}   \label{G-Nformula}
Using the  notation of Theorem \ref{MLagTheorem}, the iteration defined by
  \begin{eqnarray}\label{MultIt2}
\newvecx=\oldvecx-\left(\sum_{i=1}^ND_i^*M_{B_i}^{-1}D_i\right)^{-1}
 \left(\sum_{i=1}^ND_i^*M_{B_i}^{-1}{\bf f}_i\right).
  \end{eqnarray}
  is the Gauss-Newton iteration defined in (\ref{GaussNewtonIteration})  for the Weierstrass map ${\mathcal W}_{\vec{B}}$. Here  for $i=1, \ldots, N$ \[
{\bf f}_i:=\left(f_i(\z_1), \ldots, f_i(\z_k)\right)^T ,\;\; M_{B_i}=V_{B_i}V^*_{B_i} \text{ and } {D}_i :=
\begin{bmatrix}
  {D}_{i,1} \\
  & {D}_{i,2} \\
  & & \ddots \\
  & & & {D}_{ i,k}
\end{bmatrix}\in \C^{k\times nk}  
\]
with each block ${D}_{i,j}$ of size  $1\times n$ and defined as $
{D}_{i,j} :=\left[ \frac{\partial (f_i - p_i)}{\partial_{x_s}}(\bz_j)
\right]_{ 1\leq s\leq n}$.

\end{proposition}

\begin{proof}
$N_W$ in (\ref{GaussNewtonIteration}) uses the pseudo-inverse $J^{\dagger}$. We can expand the pseudo-inverse of $\left( \FullVandermonde^{\dagger} \left(\nabla F - \nabla \FullVandermonde \CoeffVec \right) \right)$ as follows:
{
\begin{align*}
  &\left( \FullVandermonde^{\dagger} \left(\nabla F - \nabla \FullVandermonde \CoeffVec \right) \right)^{\dagger}\\
= &\left( \left( \FullVandermonde^{\dagger} \left(\nabla F - \nabla \FullVandermonde \CoeffVec \right) \right)^{*}
         \left( \FullVandermonde^{\dagger} \left(\nabla F - \nabla \FullVandermonde \CoeffVec \right) \right)  \right)^{-1}
  \left( \FullVandermonde^{\dagger} \left(\nabla F - \nabla \FullVandermonde \CoeffVec \right) \right)^{*}\\
= &\left( \left(\nabla F - \nabla \FullVandermonde \CoeffVec \right)^{*} (VV^{*})^{-1}
         \left(\nabla F - \nabla \FullVandermonde \CoeffVec \right)  \right)^{-1}
  \left(\nabla F - \nabla \FullVandermonde \CoeffVec \right)^{*} (\FullVandermonde^{\dagger})^{*}\\
\end{align*}
}
using the fact that 
$$(\FullVandermonde^{\dagger})^{*}\FullVandermonde^{\dagger}=\left(\left(\FullVandermonde \FullVandermonde^*\right)^{-1}\right)^*\FullVandermonde\FullVandermonde^*\left(\FullVandermonde \FullVandermonde^*\right)^{-1}=
\left(\left(\FullVandermonde \FullVandermonde^*\right)^{-1}\right)^*= \left(\FullVandermonde \FullVandermonde^*\right)^{-1}.
$$

When this is substituted into (\ref{GaussNewtonIteration}) we get
{\footnotesize
\begin{align}\label{GNformula}
\newvecx &= \oldvecx
   -\left(\left( \left(\nabla F - \nabla \FullVandermonde \CoeffVec \right)^* (VV^{*})^{-1}
                   \left(\nabla F - \nabla \FullVandermonde \CoeffVec \right)  \right)^{-1}
             \left(\nabla F - \nabla \FullVandermonde \CoeffVec \right)^* (VV^{*})^{-1}
	     F\right)(\oldvecx)
\end{align}
}

To get (\ref{MultIt2}) from (\ref{GNformula}) we observe that $(\nabla V)$ is a 3-dimensional matrix of size $(kN)\times\left(\sum_{t=1}^N |B_t|\right) \times (kn)$ consisting of the $kn$
block diagonal matrices $\frac{\partial V}{\partial z_{i,j}}$ for $i=1, \ldots, k, j=1, \ldots, n$. In each block of  $\frac{\partial V}{\partial z_{i,j}}$
only one row is non-zero, the one corresponding to $\vecz_i$, and the entries of this row are changed from $b(\z_i)$ to $\frac{\partial b}{\partial x_{j}}(\vecz_i)$ for  $b$ in some $B_t$.
Since $W$ is the vector consisting of the coefficient vectors of $p_1, \ldots, p_N$ in the bases $B_1, \ldots, B_N$, we conclude that   $\nabla F - (\nabla \FullVandermonde) \CoeffVec $ is a $(kN)\times (kn)$ matrix with columns corresponding to the partial derivatives $\frac{\partial}{\partial z_{i,j}}$ ($i=1, \ldots k$, $j=1, \ldots n$), and each of these columns have $0$ entries everywhere except in the $i+(t-1)k$-th place for  $t=1, \ldots ,N$, where they are equal to $\frac{\partial (f_t-p_t)}{\partial x_j}(\vecz_i)$. To get (\ref{MultIt2}), we use the block diagonal structure of $(VV^*)^{-1}$ with blocks $M^{-1}_{B_t}$  ($t=1, \ldots, N$).

\end{proof}
\subsection{Simplified Gauss-Newton method}

In this section we describe the generalization of the univariate simplified Gauss-Newton iteration defined in Definition \ref{DefWeieriteration}. First we show  how the simplified Gauss-Newton method is obtained  from the standard Gauss-Newton method by making some adjustments based on the specifics of this particular minimization problem. Although this method does not find a minimum in the 2-norm, as we shell see, it does find a minimum that is reasonable in the context of the problem while using significantly reduced computational effort.

Consider the formula we  obtained in (\ref{GNformula}) for the Gauss-Newton iteration. What we want is to find a way to simplify this formula to a form that can be more efficiently computed.  If $\FullVandermonde$ were a square unitary matrix, the $(\FullVandermonde \FullVandermonde^{*})$ terms would be the identity matrix and would disappear from the formula. $\FullVandermonde$ is unlikely to be unitary, but it turns out that if we perform this cancellation anyway, we get a new formula that can be computed more efficiently than that of the standard Gauss-Newton, and surprisingly we still converge to a set of polynomials that can be said to be locally minimally distant from the originals---if we use a different method for measuring distance.

Dropping $(VV^{*})^{-1}$ we get
\begin{align*}
\newvecx &= \oldvecx - \left( \left( \left(\nabla F - (\nabla \FullVandermonde )\CoeffVec \right)^{*}
                    \left(\nabla F - (\nabla \FullVandermonde) \CoeffVec \right)  \right)^{-1}
             \left(\nabla F - (\nabla \FullVandermonde) \CoeffVec \right)^{*}\right)(\oldvecx) F(\oldvecx),
\end{align*}
which reduces to the \textit{simplified Gauss-Newton iteration} formula
\begin{align}
\label{SimpGNFormula}
\newvecx &= \oldvecx - \left(\nabla F -( \nabla \FullVandermonde) \CoeffVec \right)^{\dagger}(\oldvecx) F(\oldvecx).
\end{align}

In order to turn (\ref{SimpGNFormula}) into  a component-wise iteration function, as in the univariate case, we need the following assumption:
\begin{eqnarray}\label{MultAssump}
{\bf Assumption: } \;\;|B_1|=\cdots=|B_N|=k.
\end{eqnarray}

Then we can prove the following generalization of Lemma \ref{derivative},  implying the simple structure of the partial derivatives of the Weierstrass map, when expressed in terms of the Lagrange basis:

\begin{lemma}\label{PartialFormula}
Let $f \in \C^\infty_n$, $B\subset \C^\infty_n$, and assume that $|B|=k$.
  For a fixed  $\vz=(\z_1, \ldots, \z_k) \in {\mathcal R}_B$ let the Lagrange polynomials $\tL_{B,i }(\vz,\bx)$ ($i=1, \ldots, k$) defined as in Definition \ref{MLagrange}, and as before, let
   $$p(\vz, \bx) :=
 \sum_{i=1}^k f(\z_i) \tL_{B,i}\para{\vz,\bx}.$$ 
 Then
 $$
 \frac{\partial p}{\partial z_{i,j}}  (\vz,\bx)= \left(\frac{\partial (f - p)}{\partial x_j}(\z_i)\right) \tL_{B, i}(\vz, \bx).
 $$
 \end{lemma}

\begin{proof} 
The proof is similar to the proof of Lemma  \ref{derivative}, and it is based on computing the evaluations of $\frac{\partial p}{\partial z_{i,j}}$ at $\bx=\z_t$ for $t=1, \ldots,k$. Then from $|B|=k$ and $\vz\in  {\mathcal R}_B$ it follows that  $\{\tL_{B,1}, \ldots,  \tL_{B,k}\}$ generates ${\rm span}_\C B$, thus these evaluations uniquely determine the elements ${\rm span}_\C B$.
\end{proof}

Using the previous lemma we can give the following simple component-wise formula for the simplified Gauss-Newton iteration:

\begin{definition}\label{simpGN}
Let $\vec{f}=(f_1, \ldots, f_N)$ and $\vec{B}=(B_1, \ldots, B_N)$ be as
above. Let $(p_1(\vz,\bx), \ldots ,p_N(\vz,\bx))\in {\mathcal P}$ be as
in Theorem \ref{MLagTheorem}. Fix $\vz=(\bz_1, \ldots \bz_k)\in
{\mathcal R}_{\vec{B}}$. Assume that $|B_i|=k$ for all $i=1, \ldots,
N$. Define
 $$ \vec{f}_{\vz}(\bx):= \left(f_1(\bx) - p_{1}(\vz,\bx), \ldots,
 f_N(\bx) - p_{N}(\vz,\bx)\right).
 $$ Let $J_{\vz}(\bx)$ be the $N\times n$ Jacobian matrix of
 $\vec{f}_{\vz}(\bx)$.
 Assume that $\rank(J_{\vz}(\bz_i))=n$  for all $i=1, \ldots, k$.
 Then  the {\bf simplified Gauss-Newton iteration} is defined by
\begin{equation}\label{gmEWi}
\bz_i':= \bz_i- J_{\vz}(\bz_i)^\dagger\vec{f}(\bz_i) \;\;i=1, \ldots,
k.
\end{equation}
\end{definition}

The following theorem is a generalization of Theorem \ref{UniSimpTheorem} and asserts that $\vz\in (\C^n)^k$ is a fixed point of the simplified Gauss-Newton iteration if it corresponds to perturbation functions which are locally pointwise minimal. 

\begin{theorem}\label{MultSimpTheorem}
A point $\vz=(\z_1, \ldots, \z_k)\in {\mathcal R}_{\vec{B}}$ is a fixed point of the  simplified Gauss-Newton iteration in (\ref{gmEWi}) if there exists an open neighborhood $U$ of $\vz$ such that for all $\tilde{\vz}=(\tilde{\z}_1, \ldots, \tilde{\z}_k)$ and $\vz'=(\z'_1, \ldots, \z'_k)$ in $U$ and for all $i=1, \ldots, k$ 
\begin{eqnarray}\label{MultPointwise}
|p_1(\vz, \tilde{\z}_i)|^2 + \cdots + |p_1(\vz, \tilde{\z}_i)|^2\leq |p_1(\vz', \tilde{\z}_i)|^2 + \cdots + |p_1(\vz', \tilde{\z}_i)|^2.
\end{eqnarray}
Note that this includes the case when $\z_1, \ldots, \z_k$ are common roots of $f_1, \ldots, f_N$, in which case $p_1(\vz, \bx)= \cdots=p_1(\vz, \bx)=0$.
\end{theorem}

\begin{proof}
For $\vz$ to be a fixed point for the simplified Gauss-Newton iteration, it is sufficient to prove that for all $i=1, \ldots, k$ $J_{\vz}(\bz_i)^*\vec{f}(\bz_i)=0$, which is equivalent to
\begin{eqnarray}\label{FixpointClaim}
\sum_{t=1}^N \frac{\partial(f_t - p_t)}{\partial x_j} (\z_i) \overline{f_t(\z_i)}=0 \text{ for all }  i=1, \ldots, k,\;j=1, \ldots, n.
\end{eqnarray}
By (\ref{MultPointwise}) we have that for any  $\tilde{\vz}=(\tilde{\z}_1, \ldots, \tilde{\z}_k)\in U$ and for all 
$i=1, \ldots, k$ and $j=1, \ldots, n$
$$
\frac{\partial }{\partial z_{i,j}}\left(|p_1(\vz, \tilde{\z}_i)|^2 + \cdots + |p_1(\vz, \tilde{\z}_i)|^2 \right)=0,
$$
which implies that
$$
\sum_{t=1}^N \overline{p_t(\vz, \tilde{\z}_i)}\frac{\partial p_t(\vz, \tilde{\z}_i)}{\partial z_{i,j}}=0.
$$
Using Lemma \ref{PartialFormula} we get that 
$$
\sum_{t=1}^N \overline{p_t(\vz, \tilde{\z}_i)} \left(\frac{\partial (f_t - p_t)}{\partial x_j}(\z_i)\right) \tL_{B, i}(\vz, \tilde{\z}_i)=0.
 $$
 Approaching with $\tilde{\vz}$ to $\vz$ we get (\ref{FixpointClaim}).
\end{proof}

\subsection{Quadratic Iteration}

The quadratic iteration method explicitly calculates the Gradient and Hessian of the function $\CoeffVec^{*}\CoeffVec$, evaluates them at the current point $\oldvecx$, and directly solves for the critical point $\vecz$ using the linear system 
\begin{align}
\label{QuadraticInterpolationFormula}
0
= G(\vecz_{0}, \veczconj_{0})
+ H(\vecz_{0}, \veczconj_{0})
\left[
\begin{array}{c}
\vecz - \vecz_0\\
\veczconj - \veczconj_0
\end{array}
\right].
\end{align}

Since such a calculated critical point is as likely to be a maximum (or saddle) point as a minimum, usage of $H$ is adjusted by removing positive eigenvalues to ensure movement towards a desired minimum. Additionally, if the 2-norm of $\CoeffVec$ at $\newvecx$ is greater than that at $\oldvecx$, points along the line between $\newvecx$ and $\oldvecx$ closer and closer to $\oldvecx$ are tested until a decrease in the norm is detected.

\subsection{Conjugate Gradient method }

The conjugate gradient method does repeated one dimensional minimizations, in a single direction for each iteration, until a local minimum in all directions is found. We will label the directions used for iteration $i$ as $g_i$. These directions are not chosen randomly, but in such a way as to find the minimum in as few iterations as possible. Below we will show that for quadratic functions the minimum will be found in at most $n$ iterations.\\
 
A general quadratic function in $n$ variables has the form
\begin{align*}
Q(\vecz) = K + \vecz^{T} L + \vecz^{T} M \vecz
\end{align*}
where $K$ is a scalar, $L$ is an $n$ dimensional vector, and $M$ is an $n \times n$ matrix. This method assumes that the function is real-valued and that the quadratic term, $\vecz^{T} M \vecz$, is nonnegative for all $\vecz$. Otherwise it does not make sense to talk of the function's minimum.\\

To simplify the discussion, we will work with a translated version of this quadratic function, {$Q(\vecz + \vecz_0) - Q(\vecz_0)$}, so that the starting point of the iteration is at the origin, and the value of the function at the origin is zero. We can then assume that our quadratic function $Q$ has the form
\begin{align*}
Q &=  \vecz^{T} L + \vecz^{T} M \vecz.
\end{align*}

The key idea behind the conjugate gradient method is to choose each iteration's search direction $g_i$ to be \textit{conjugate} to the previous directions, which means that
\begin{align*}
g_{i}^{T} M g_{j} = 0 \ \forall j < i.
\end{align*}
Then for any linear combination of these conjugate directions, $\sum_{j=1}^{n}a_{j} g_{j}$ with $a_{j} \in \C$, we have
\begin{align*}
Q\left( \sum_{j=1}^{n}a_{j} g_{j} \right)
&= \sum_{j=1}^{n} \left( a_{j} g_{j}^{T} \right) L
+  \sum_{j=1}^{n} \left( a_{j} g_{j}^{T} \right)
M  \sum_{j=1}^{n} \left( a_{j} g_{j} \right)\\
&= \sum_{j=1}^{n} \left( a_{j} g_{j}^{T} L + a_{j}
                         g_{j}^{T} M a_{j} g_{j} \right)\\
&= \sum_{j=1}^{n} Q\left( a_{j} g_{j} \right).
\end{align*}
So minimization of $Q$ can occur independently in each of the conjugate directions. It must be complete after $n$ iterations since all possible search directions will have been exhausted.\\

By using calculated gradient information, the optimization directions are each chosen to be as close to the direction of steepest descent of the function as possible while maintaining the required conjugacy relationship. This allows the method to stop in fewer iterations when there are some directions that are already at or near a minimum.\\

When minimizing functions that are not precisely quadratic, such as the problem we are dealing with, the exact solution is not guaranteed to be found within $n$ iterations since the effects of the $g_i$ vectors on the value of the function are not independent. However, the practice of following the steepest conjugate directions first can still allow us to come acceptably close to the solution within $n$ iterations depending on the characteristics of our function and our required tolerance.\\

It should be noted that conjugate directions can be calculated without using the matrix $M$. This saves significant computation time by avoiding calculation of the Hessian which would otherwise be required when using a quadratic Taylor series approximation.

\section{Algorithmic Complexity}

The per iteration operation counts are represented in the following table, where\\
$N$ is the number of input (and output) functions;\\
$n$ is the number of variables used in the input functions;\\
$k$ is the number of input (and output) roots;\\
$\beta$ is the number of bits of accuracy used for the intermediate steps of the conjugate gradient method\\
Here we make the assumption that all perturbation bases $B_1, \ldots, B_N$ has cardinality $k$.

\noindent
\begin{tabular}[t]{|l|c|c|c|}
\hline
       & Input       & Basis       & Arithmetic \\
Method & Evaluations & Evaluations & Operations \\
\hline
Simp G-N
& $\mathcal{O} (N \cdot k \cdot n)$
& $\mathcal{O} (N \cdot k^2 \cdot n)$
& $\mathcal{O} ($max$(N \cdot k^3, N \cdot k \cdot n^2))$\\
\hline
Std G-N
& $\mathcal{O} (N \cdot k \cdot n)$
& $\mathcal{O} (N \cdot k^2 \cdot n)$
& $\mathcal{O} (N \cdot k^3 \cdot n^2)$\\
\hline
Quad It
& $\mathcal{O} (N \cdot k \cdot n^2) $
& $\mathcal{O} (N \cdot k^2 \cdot n^2) $
& $\mathcal{O} (N \cdot k^3 \cdot n^2) $\\
\hline
Conj Grd
& $\mathcal{O} (N \cdot k \cdot (n + \beta)) $
& $\mathcal{O} (N \cdot k^2 \cdot (n + \beta)) $
& $\mathcal{O} (N \cdot k^3 \cdot (n + \beta)) $\\
\hline
\end{tabular}
\vspace{0.5cm}

The \textit{Input Evaluations} column is the number of evaluations of input functions or their derivatives.
The \textit{Basis Evaluations} column is the number of evaluations of perturbation basis functions or their derivatives.
The \textit{Arithmetic Operations} column is the number of simple scalar arithmetic operations, excluding the operations involved in evaluating the functions from the preceding two columns.\\

Calculation of the gradient of our 2-norm requires evaluation of $n$ partial derivatives at each of $k$ input roots for each of the $N$ input functions and $N$ perturbation functions, for a total of $N \cdot k \cdot n$  evaluations and, since each perturbation functions are the sum of $k$ basis functions,  $N \cdot k^{2} \cdot n$ basis evaluations. This accounts for the $N \cdot k \cdot n$ input evaluations and $N \cdot k^{2} \cdot n$ basis evaluations for the two Gauss-Newton methods and the conjugate gradient method.\\

The number of function evaluations for the quadratic iteration method is dominated by the calculation of the Hessian matrix, which the other methods avoid. The Hessian requires evaluation at $n^{2}$ partial derivatives for each of $N$ input functions and $N$ perturbation functions at $k$ different points. This is a factor of $n$ more evaluations than is required by the gradient calculation, giving us $N \cdot k \cdot n^{2}$ input evaluations and $N \cdot k^{2} \cdot n^{2}$ basis function evaluations. \\

Each method starts by calculating the basis function coefficients for the perturbation function at the current iteration point. This requires the solution of $N$ different linear systems. Since the Vandermonde matrices  have dimension $k \times k$, each of this steps requires $\mathcal{O}(k^{3})$ operations.\\

Furthermore, the Simplified Gauss-Newton method requires $\mathcal{O}(N \cdot n^{2})$ operations to solve each of the $k$ equations in formula (\ref{gmEWi}). This requires effort $\mathcal{O}(N \cdot k \cdot n^{2})$. This may be greater or less than the effort to solve the above Vandermonde system, so the complexity is determined to be the greater of $\mathcal{O}(N \cdot k^{3})$ and $\mathcal{O}(N \cdot k \cdot n^{2})$. If the solution of the Vandermonde system is the dominating factor, further savings can be realized if all of the input functions use the same perturbation basis. The complexity is then the greater of $\mathcal{O}(k^{3})$ and $\mathcal{O}(N \cdot k \cdot n^{2})$.\\

The standard Gauss-Newton method requires $\mathcal{O}( k^{3} n^{3})$ operations to solve equation (\ref{MultIt2}) since the matrix to be inverted  is a $n k \times n k$ matrix, plus $\mathcal{O}(Nk^2n^2)$ additions to compute the sum of $N$  matrices each of size $n k \times n k$. These can be bounded by $\mathcal{O}( Nk^{3} n^{2})$ since $N>n$.\\

The quadratic iteration method requires $N \cdot k$ operations to calculate each entry of the $n k \times n k$ Hessian matrix. This is because each of $N$ perturbation functions contributes to every matrix entry and there are $k$ basis function evaluations that need to be combined to get each perturbation function evaluation. Solution of the linear system (\ref{QuadraticInterpolationFormula}) involving this matrix requires $\mathcal{O}(k^{3} \cdot n^{3})$ operations. Since $N > n$, it is the setup of the Hessian that dominates, which requires $\mathcal{O}(N \cdot k^{3} \cdot n^{2})$ operations.\\

The $\beta$ factor for the conjugate gradient method comes from the line minimization performed during each step. The method assumes that the directional derivative along the line is zero at the minimum. Thus the more accurate the minimization, the more accurate this assumption. The factor of $\beta$ is the average number of steps to arrive at this minimization to machine precision.  Some functions' line minimums are found more rapidly than this and for some functions less precision can be used without sacrificing convergence rate.\\

In most cases the simplified Gauss-Newton method does the fewest operations per iteration by a factor of $k$. For some problems (i.e. where $n^2 > k^2{\cdot}(n+\beta)$) the conjugate gradient method appears that it would provide better performance. Tests indicate that for problems this complicated the conjugate gradient method is unlikely to converge to a good local minimum (i.e. a minimum close to the global minimum), so using the simplified Gauss-Newton would still be the recommended method. Although the quadratic iteration and standard Gauss-Newton methods have the same reported number of operations per iteration, quadratic iteration is actually a nontrivial constant factor slower than the standard Gauss -Newton method.

\section{Comparison Tests}

\subsection{Test Design}

Tests were performed using four different configuration. The configurations differed in the numbers of polynomials ($N$), variables ($n$), degrees ($D$), and number of common roots ($k$) for which to search.\\

Each random polynomial was generated by creating all monomials of total degree less than or equal to $D$, the degree chosen for that problem, then applying a randomly generated coefficient between $-100$ and $100$. $k$ random points were then chosen in the range $(-10, 10)$. Polynomials were then generated that interpolated each of these random polynomials at each of the random points. These interpolating polynomials were subtracted from the original random polynomials to give a system with $k$ common roots that are referred to as the unperturbed polynomials.\\

A perturbation basis ($B$) was chosen using $k$ monomials of smallest total degree.
The input polynomials were generated from these unperturbed polynomials by adding to each polynomial a randomly generated polynomial with terms chosen from the perturbation basis.
Each of these randomly generated polynomials is created as $\sum^{k}_{i=1} r_i \cdot B_i$, where each $r_i$ is a different randomly generated number and $B_i$ is the $i$th element of the perturbation basis $B$. For each set of tests, $r_i$ was chosen in the five different ranges $(-10^x, 10^x)$ for $x \in \left\{ -2, -1, 0, 1, 2 \right\}$. Ten problems were run for each range, making a total of fifty problems per configuration. The starting point for each iteration was chosen as the roots of the unperturbed polynomials, modified by adding a vector randomly chosen within the unit hypersphere.

{
\setlength{\parindent}{0cm}
\setlength{\parskip}{0cm}

\subsection{Tables}
\nopagebreak
\hspace{-0.1cm}
\begin{minipage}{\textwidth}
\begin{footnotesize}
\begin{tabular}[t]{l|c|r|r|r|r|r|r|r|r|r|r|r|r}
 & \multicolumn{1}{|c}{\makebox[0.8cm]{\%\nolinebreak Con-}} & \multicolumn{3}{|c}{Rel Residual} & \multicolumn{2}{|c}{Abs Resid} & \multicolumn{3}{|c}{Rel Output Norm} & \multicolumn{3}{|c}{Abs Output Norm} & \multicolumn{1}{|c}{Iter} \\
Method& \multicolumn{1}{|c}{\makebox[0.8cm]{verged}} & \multicolumn{1}{|c}{Min} & \multicolumn{1}{|c}{Avg} & \multicolumn{1}{|c}{Max} & \multicolumn{1}{|c}{Min} & \multicolumn{1}{|c}{Max} & \multicolumn{1}{|c}{Min} & \multicolumn{1}{|c}{Avg} & \multicolumn{1}{|c}{Max} & \multicolumn{1}{|c}{Min} & \multicolumn{1}{|c}{Avg} & \multicolumn{1}{|c}{Max} & \multicolumn{1}{|c}{Cnt} \\
\hline
Simp G-N & 100 & 1.00 & 1.00 & 1.00 & 4.7e-7 & 0.80 & 1.00 & 1.00 & 1.00 & 2.9e-4 & 0.04 & 0.34 & 4.42\\
\hline
Std G-N & 100 & 1.00 & 1.00 & 1.00 & 4.7e-7 & 0.80 & 1.00 & 1.00 & 1.00 & 2.9e-4 & 0.04 & 0.34 & 4.42\\
\hline
Quad It & 100 & 1.00 & 1.00 & 1.00 & 4.7e-7 & 0.80 & 1.00 & 1.00 & 1.00 & 2.9e-4 & 0.04 & 0.34 & 4.98\\
\hline
Conj Grd & 100 & 1.00 & 1.00 & 1.00 & 4.7e-7 & 0.80 & 1.00 & 1.00 & 1.00 & 2.9e-4 & 0.04 & 0.34 & 4.20\\
\hline
\end{tabular}

5 polynomials of degree 3 in 1 variable with 1 common root.\\
There were 50 problems for which all methods converged.
\vspace{7mm}

\begin{tabular}[t]{l|c|r|r|r|r|r|r|r|r|r|r|r|r}
 & \multicolumn{1}{|c}{\makebox[0.8cm]{\%\nolinebreak Con-}} & \multicolumn{3}{|c}{Rel Residual} & \multicolumn{2}{|c}{Abs Resid} & \multicolumn{3}{|c}{Rel Output Norm} & \multicolumn{3}{|c}{Abs Output Norm} & \multicolumn{1}{|c}{Iter} \\
Method& \multicolumn{1}{|c}{\makebox[0.8cm]{verged}} & \multicolumn{1}{|c}{Min} & \multicolumn{1}{|c}{Avg} & \multicolumn{1}{|c}{Max} & \multicolumn{1}{|c}{Min} & \multicolumn{1}{|c}{Max} & \multicolumn{1}{|c}{Min} & \multicolumn{1}{|c}{Avg} & \multicolumn{1}{|c}{Max} & \multicolumn{1}{|c}{Min} & \multicolumn{1}{|c}{Avg} & \multicolumn{1}{|c}{Max} & \multicolumn{1}{|c}{Cnt} \\
\hline
Simp G-N & 98 & 1.00 & 1.00 & 1.00 & 1.3e-5 & 0.85 & 1.00 & 1.00 & 1.00 & 9.5e-4 & 0.05 & 0.28 & 4.67\\
\hline
Std G-N & 100 & 1.00 & 1.16 & 1.45 & 1.4e-5 & 0.97 & 0.64 & 0.88 & 1.02 & 8.1e-4 & 0.04 & 0.22 & 4.86\\
\hline
Quad It & 100 & 1.00 & 1.16 & 1.45 & 1.4e-5 & 0.97 & 0.64 & 0.88 & 1.02 & 8.1e-4 & 0.04 & 0.22 & 8.08\\
\hline
Conj Grd & 100 & 1.04 & 957 & 1.2e4 & 0.04 & 1.02 & 0.70 & 14.5 & 87.4 & 0.03 & 0.09 & 0.22 & 25.96\\
\hline
\end{tabular}

5 polynomials of degree 2 in 2 variables with 2 common roots.\\
There were 49 problems for which all methods converged.
\vspace{7mm}

\begin{tabular}[t]{l|c|r|r|r|r|r|r|r|r|r|r|r|r}
 & \multicolumn{1}{|c}{\makebox[0.8cm]{\%\nolinebreak Con-}} & \multicolumn{3}{|c}{Rel Residual} & \multicolumn{2}{|c}{Abs Resid} & \multicolumn{3}{|c}{Rel Output Norm} & \multicolumn{3}{|c}{Abs Output Norm} & \multicolumn{1}{|c}{Iter} \\
Method& \multicolumn{1}{|c}{\makebox[0.8cm]{verged}} & \multicolumn{1}{|c}{Min} & \multicolumn{1}{|c}{Avg} & \multicolumn{1}{|c}{Max} & \multicolumn{1}{|c}{Min} & \multicolumn{1}{|c}{Max} & \multicolumn{1}{|c}{Min} & \multicolumn{1}{|c}{Avg} & \multicolumn{1}{|c}{Max} & \multicolumn{1}{|c}{Min} & \multicolumn{1}{|c}{Avg} & \multicolumn{1}{|c}{Max} & \multicolumn{1}{|c}{Cnt} \\
\hline
Simp G-N & 70 & 1.00 & 1.00 & 1.00 & 1.9e-5 & 0.31 & 1.00 & 1.00 & 1.00 & 2.5e-3 & 0.05 & 0.24 & 6.69\\
\hline
Std G-N & 76 & 1.14 & 1.62 & 2.64 & 2.4e-5 & 0.67 & 0.30 & 0.40 & 0.51 & 7.7e-4 & 0.02 & 0.12 & 6.31\\
\hline
Quad It & 92 & 1.14 & 15.0 & 427 & 2.4e-5 & 0.69 & 0.30 & 0.69 & 8.31 & 7.7e-4 & 0.02 & 0.11 & 29.29\\
\hline
Conj Grd & 90 & 2.65 & 3.5e3 & 4.6e4 & 0.30 & 0.91 & 0.69 & 21.6 & 94.0 & 0.13 & 0.18 & 0.25 & 19.97\\
\hline
\end{tabular}

5 polynomials of degree 2 in 4 variables with 6 common roots.\\
There were 35 problems for which all methods converged.
\vspace{7mm}

\begin{tabular}[t]{l|c|r|r|r|r|r|r|r|r|r|r|r|r}
 & \multicolumn{1}{|c}{\makebox[0.8cm]{\%\nolinebreak Con-}} & \multicolumn{3}{|c}{Rel Residual} & \multicolumn{2}{|c}{Abs Resid} & \multicolumn{3}{|c}{Rel Output Norm} & \multicolumn{3}{|c}{Abs Output Norm} & \multicolumn{1}{|c}{Iter} \\
Method& \multicolumn{1}{|c}{\makebox[0.8cm]{verged}} & \multicolumn{1}{|c}{Min} & \multicolumn{1}{|c}{Avg} & \multicolumn{1}{|c}{Max} & \multicolumn{1}{|c}{Min} & \multicolumn{1}{|c}{Max} & \multicolumn{1}{|c}{Min} & \multicolumn{1}{|c}{Avg} & \multicolumn{1}{|c}{Max} & \multicolumn{1}{|c}{Min} & \multicolumn{1}{|c}{Avg} & \multicolumn{1}{|c}{Max} & \multicolumn{1}{|c}{Cnt} \\
\hline
Simp G-N & 94 & 1.00 & 1.00 & 1.00 & 3.1e-5 & 0.73 & 1.00 & 1.00 & 1.00 & 1.9e-3 & 0.06 & 0.33 & 8.05\\
\hline
Std G-N & 98 & 0.95 & 1.18 & 1.46 & 3.6e-5 & 0.87 & 0.39 & 0.54 & 0.72 & 1.0e-3 & 0.03 & 0.16 & 5.79\\
\hline
Quad It & 100 & 0.95 & 2.17 & 42.8 & 3.6e-5 & 0.87 & 0.39 & 0.61 & 3.70 & 1.0e-3 & 0.03 & 0.16 & 20.38\\
\hline
Conj Grd & 90 & 1.07 & 2.3e3 & 2.7e4 & 0.25 & 0.95 & 0.57 & 24.5 & 118 & 0.11 & 0.20 & 0.31 & 22.93\\
\hline
\end{tabular}

9 polynomials of degree 2 in 4 variables with 6 common roots.\\
There were 42 problems for which all methods converged.
\vspace{7mm}

\end{footnotesize}
\end{minipage}

}

\subsection{Explanation of Tables}

The first column of the tables names the method used in the test. \textit{Simp G-N} is the simplified Gauss-Newton, \textit{Std G-N} is the standard Gauss-Newton method, \textit{Quad It} is the quadratic iteration method, and \textit{Conj Grad} is the conjugate gradient method.\\

All calculated values except the convergence percentage are measuring only the results from the problems for which all methods converged. This way we ensure that the numbers from each method are comparable.\\

The \textit{Converge \%} column indicates the percentage of problems for which the method converged. For these tests, a method is said to have converged if within 128 iterations the change produced during each of two consecutive iterations is less than 0.001.  For the Gauss-Newton type methods, if three consecutive iterations have increasing step size, the method is considered to be diverging. The quadratic iteration and conjugate gradient methods are designed such that each step guaranteed to move closer to the desired local minimum so no divergence test is done.\\

The following three columns report a relative residual, where residual is the 2 norm of the vector with entries equal to the input polynomials substituted at each output root. For each method the residual is divided by the residual calculated for the Simplified G-N method to get a relative residual that will be less sensitive to the scaling of the individual test problems. It also allows for easy comparison with the Simplified G-N method. By definition then this value will be precisely 1.0 for the Simplified G-N method. The three columns report the minimum, arithmetic mean, and maximum of this relative residual among all the convergent test cases.\\

The next two columns report the minimum and maximum residual calculated for the sample problems. These are not scaled relative to the Simplified G-N result. A smaller value here suggests that the output roots are closer to being roots of the input polynomials. A value less than one suggests that the output roots are closer to being roots of the original system than the input roots.\\

The \textit{Abs Output Norm} columns report the minimum, mean, and maximum absolute output norm, i.e. the 2-norm of the coefficients of the perturbation functions. A smaller value means the output polynomials have coefficients closer to those of the input polynomials.\\

The \textit{Rel Output Norm} columns report the minimum, mean, and maximum relative output norm. Values smaller than $1.0$ indicate a smaller (better) absolute output norm than the Simplified G-N method.\\

The \textit{Iter Cnt} column reports the average number of iterations required until convergence is achieved.

\bibliographystyle{abbrv}
\def\cprime{$'$} \def\cprime{$'$} \def\cprime{$'$}

\end{document}